\documentclass[11pt]{article}
\usepackage[letterpaper,margin=1in]{geometry}
\usepackage[T1]{fontenc}
\usepackage{lmodern}
\usepackage{amsmath,amssymb,amsthm}
\usepackage{microtype}
\usepackage{enumitem}
\usepackage{booktabs}
\usepackage{tikz}
\usepackage{url}
\usetikzlibrary{arrows.meta,positioning,decorations.pathreplacing}
\usepackage[colorlinks=true,linkcolor=blue!55!black,citecolor=blue!55!black,
  urlcolor=blue!55!black]{hyperref}
\newtheorem{theorem}{Theorem}[section]
\newtheorem{lemma}[theorem]{Lemma}

\theoremstyle{definition}

\newcommand{\Path}{\operatorname{Path}}
\setlist[enumerate]{topsep=5pt,itemsep=5pt,leftmargin=*}
\setlist[itemize]{topsep=5pt,itemsep=5pt,leftmargin=*}
\allowdisplaybreaks[1]
\hypersetup{pdftitle={Improved Integrality Gap for Multicommodity Flow on Trees},
  pdfauthor={Elfarouk Harb},
  pdfsubject={Weighted unit-demand multicommodity flow on trees}}
\title{Improved Integrality Gap for Multicommodity Flow on Trees}
\author{Elfarouk Harb\thanks{This proof sat in a drawer for more than five years. It grew out of my pre-AI work on the problem with Professor Chandra Chekuri during my first year of the PhD, whom I thank for his valuable comments on this manuscript. I set it aside for an embarrassing reason: Lemma~\ref{lem:packing} seemed ``clearly'' false to me. Working from my old private notes, ChatGPT Astra proved the lemma in half an hour and repaired several minor correctness issues along the way. It also improved my original integrality-gap bound from $4/11$ to $2/5$ via a cleverer coloring scheme.}\\eyfmharb@gmail.com}
\date{}

\begin{document}
\maketitle
\vspace{-2em}
\begin{abstract}
We improve the best known lower bound on the integrality gap for weighted unit-demand multicommodity flow on trees from \(1/4\) to \(2/5\), improving on the long-standing bound of Chekuri, Mydlarz, and Shepherd~\cite{CMS}. We give the proof in two stages. First, a surprisingly simple packing lemma and an inductive coloring argument give an intermediate bound of \(4/11\). We then refine the argument to obtain \(2/5\).
\end{abstract}

\section{Introduction and main result}
\label{sec:introduction}

Let \(T=(V,E)\) be a tree, and let \(u_e\ge1\) be an integer capacity for each edge \(e\). Let \(F\) be a finite collection of labeled requests. Each request \(f\in F\) has a nonnegative weight \(w_f\) and an associated path \(\Path(f)\), the unique path between its endpoints on the tree. Different requests may have the same endpoints; each request can be accepted at most once. For a subset \(S\) of requests, write
\[
  w(S)=\sum_{f\in S}w_f.
\]
A subset \(S\) is \emph{feasible} if
\[
  |\{f\in S:e\in\Path(f)\}|\le u_e
  \qquad\text{for every edge }e.
\]
We call the cardinality on the left the \emph{load} of \(S\) on \(e\). When we work directly with a collection of paths, its load on an edge likewise means the number of its paths that use that edge.

The natural linear programming relaxation for unit-demand multicommodity flow on trees assigns an acceptance fraction \(x_f\) to each request:
\begin{equation}
\begin{aligned}
  \text{maximize}\quad &\sum_{f\in F}w_fx_f,\\
  \text{subject to}\quad
  &\sum_{\substack{f\in F : e\in\Path(f)}}x_f\le u_e
      &&\text{for every edge }e,\\
  &0\le x_f\le1 &&\text{for every request }f.
\end{aligned}
\label{eq:lp}
\end{equation}

The \emph{integrality gap} is the infimum, over instances with positive LP optimum, of the ratio of the optimal integer solution value to the optimal LP value. We prove our result in two stages. The first stage gives the following \(4/11\) bound.

\begin{theorem}\label{thm:eleven-gap}
For every feasible solution \(x\) of~\eqref{eq:lp}, there is a feasible subset \(S \subseteq F\) of requests such that
\[
  w(S)\ge\frac{4}{11}\sum_{f\in F}w_fx_f.
\]
Moreover, for rationally encoded inputs, such a subset can be found in polynomial time.
\end{theorem}

We keep the complete proof of this intermediate result because it is elementary and introduces the binning and contraction coloring argument. Afterward, we tighten one step in the argument, and prove our final result.

\begin{theorem}\label{thm:main}
For every feasible solution \(x\) of~\eqref{eq:lp}, there is a feasible subset \(S\subseteq F\) of requests such that
\[
  w(S)\ge\frac{2}{5}\sum_{f\in F}w_fx_f.
\]
Consequently, the integrality gap is at least \(2/5\). Moreover, for rationally encoded inputs, such a subset can be found in polynomial time.
\end{theorem}

Chekuri, Mydlarz, and Shepherd~\cite{CMS} proved that the integrality gap of this LP is at least \(1/4\). K\"onemann, Parekh, and Pritchard~\cite{KPP} improved this bound to \(1/3\) when every capacity is at least two, and obtained bounds approaching one as the minimum capacity grows.

\paragraph{Related work.} Garg, Vazirani, and Yannakakis~\cite{GVY} showed that the cardinality version is NP-hard and MAX SNP-hard on trees, although it can be solved exactly when every edge has capacity one. They also gave a primal--dual algorithm that produces an integral routing and a multicut, with the routing value at least half the total capacity of the multicut.

Cheriyan, Jord\'an, and Ravi~\cite{CJR} subsequently proved that every half-integral LP solution---that is, one with \(x_f\in\{0,\tfrac{1}{2},1\}\)---can be rounded while retaining at least \(2/3\) of its weight. The same fraction is also an upper bound on the integrality gap. On a three-leaf star with unit capacities and one request of weight one between each pair of leaves, assigning \(x_f=\tfrac{1}{2}\) to each request gives LP value \(3/2\), whereas the integral optimum is one. It is conjectured that the integrality gap is exactly \(2/3\).

For arbitrary nonnegative weights, Chekuri, Mydlarz, and Shepherd (CMS)~\cite{CMS} obtained the first constant lower bound on the integrality gap, \(1/4\), through a stronger \(4k\)-coloring theorem for paths in trees. K\"onemann, Parekh, and Pritchard~\cite{KPP} later gave an objective-preserving rounding with additive congestion at most two. In particular, if \(\mu=\min_{e\in E}u_e\ge2\), their results show that the integrality gap is at least
\[
  \max\left\{\frac{1}{3},\;
  \frac{\mu(\mu-1)}{(\mu+1)(\mu+2)}\right\}.
\]
Thus, before the present work, the best known lower and upper bounds on the integrality gap for arbitrary weights and arbitrary capacities were \(1/4\) and \(2/3\), respectively.

The tree decomposition underlying the CMS bound has also found applications beyond trees: Naves, Shepherd, and Xia~\cite{NSX} used it as a building block in their \(1/224\)-approximation for maximum-weight disjoint paths in capacitated outerplanar graphs. Separately, Pan and Goemans~\cite{PanGoemans} studied bicriteria algorithms for demand matching, which is the star special case with arbitrary demands.

\subsection*{The idea of the proof}

Let \(W=\sum_f w_fx_f\) be the weight of the fractional solution. Our starting point is a bicriteria rounding theorem of K\"onemann, Parekh, and Pritchard~\cite{KPP}: if we allow the load on each edge to exceed its capacity by at most two, we can select a subset of requests with total weight at least \(W\).

We apply this rounding theorem to a copied instance with scaled capacities. Make four \emph{copies} of every request and multiply every capacity by four. The resulting fractional objective value is \(4W\). Rounding then gives selected copies of total weight at least \(4W\), with at most \(4u_e+2\) selected copies whose paths use edge \(e\). We will partition the selected copies into eleven color classes satisfying two conditions:

\begin{enumerate}
    \item Each color class respects the original capacities \(u_e\); that is, each color class is feasible.
    \item Each color class contains at most one copy of any original request.
\end{enumerate}
The second condition prevents double counting when we return to the original instance. For example, suppose an original request has weight \(w_f=7\). If a color class contains two copies of this request, then its weight in the copied instance is \(14\), but those copies correspond to a single request of weight \(7\) in the original instance. Restricting each color class to at most one copy of every original request ensures that the two weights agree.

Thus, the heaviest color class has weight at least \(4W/11\). Section~\ref{sec:rounding} explains how to enforce the conditions and gives the full reduction.

The main work is the division into eleven colors. We first solve the coloring problem on a star. There a simple binning and greedy coloring works. To handle a general tree, we recurse, by removing a star at the end of the tree, color the smaller tree recursively, and then restore the remaining paths. The issue is to retain enough information about the old bins from the recursive call so that the colors still work when the star is restored.

We first carry out this argument exactly as stated and obtain the \(4/11\) bound. In Section~\ref{sec:ten-colors}, we refine the argument further and show that the paths can instead be colored together with only ten colors, improving the result to $4W/10=2W/5$. 

\section{Reduction to a coloring problem}
\label{sec:rounding}

We first restate the rounding theorem in the form needed here.

\begin{theorem}[K\"onemann--Parekh--Pritchard~\cite{KPP}]
\label{thm:kpp}
Consider unit-demand requests with nonempty paths in a tree with positive integer capacities \(c_e\) and nonnegative weights \(w_f\). Let \(x\) be a fractional solution satisfying
\[
  0\le x_f\le1\quad\text{for every }f,
  \qquad
  \sum_{f:e\in\Path(f)}x_f\le c_e\quad\text{for every }e.
\]
One can find, in polynomial time, a subset \(J\) of the requests such that
\[
  \sum_{f\in J}w_f\ge\sum_f w_fx_f
  \qquad\text{and}\qquad
  |\{f\in J:e\in\Path(f)\}|\le c_e+2
  \quad\text{for every }e.
\]
\end{theorem}
The proof appears in Section~2 of~\cite{KPP}. Because requests are labeled, the theorem also applies when several requests use the same path.

\paragraph{Encoding the copy restriction as an edge capacity.} One of the conditions that we must ensure is that each color class contains at most one copy of any original request. We encode this restriction using capacity-one edges. For each original request \(f\) with endpoints \(s\) and \(t\), attach a new leaf to \(s\) and another to \(t\), assigning capacity one to both new edges. Extend \(f\) so that its endpoints are these two new leaves; see Figure~\ref{fig:private}.

Before copying, the two new edges are used only by \(f\); after copying, they are used only by the four copies of \(f\). Even if two original requests have identical endpoints, we give them separate pairs of new leaves.

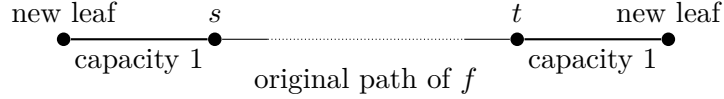
\begin{figure}[htb]
\centering
\begin{tikzpicture}[x=1cm,y=1cm,
  dot/.style={circle,fill,inner sep=1.8pt}]
  \node[dot,label=above:{new leaf}] (a) at (0,0) {};
  \node[dot,label=above:{$s$}] (s) at (2,0) {};
  \node[dot,label=above:{$t$}] (t) at (6,0) {};
  \node[dot,label=above:{new leaf}] (b) at (8,0) {};
  \draw[thick] (a)--node[below] {capacity $1$}(s);
  \draw (s)--(2.7,0);
  \draw[densely dotted] (2.7,0)--(5.3,0);
  \draw (5.3,0)--(t);
  \draw[thick] (t)--node[below] {capacity $1$}(b);
  \node at (4,-0.55) {original path of $f$};
\end{tikzpicture}
\caption{One pair of new leaves is added for each original request. All copies of that request share the two new capacity-one edges.}
\label{fig:private}
\end{figure}

The extended requests retain their weights and traverse the same original edges as before. Each new edge is used only by request \(f\), so the corresponding capacity constraint is exactly \(x_f\le1\). Hence the construction preserves every feasible fractional solution and its value. If the endpoints of a request coincide, \(s=t\), we attach two distinct leaves at that vertex; its extended path is then a nonempty two-edge path.

Sharing the private edges among the copies turns the one-copy-per-color-class restriction into an edge-capacity constraint. Fix a feasible solution \(x\) of~\eqref{eq:lp}, and let \(W=\sum_f w_fx_f\). Let \(u_e\) denote the unscaled capacity of every edge in the enlarged tree, including \(u_e=1\) on each private edge. Create four labeled copies of each extended request, each with weight \(w_f\) and fractional value \(x_f\), and scale every edge capacity to \(4u_e\). The resulting fractional solution is feasible because, for every edge \(e\),
\[
  \sum_{f:e\in\Path(f)}(x_f+x_f+x_f+x_f)
  =4\sum_{f:e\in\Path(f)}x_f
  \le4u_e.
\]
Its fractional objective value is \(4W\). Applying Theorem~\ref{thm:kpp} yields selected copies of total weight at least \(4W\), with at most \(4u_e+2\) selected copies whose paths use edge \(e\). We recover feasibility using the following coloring theorem, proved in Section~\ref{sec:trees}. For the coloring argument, we identify each selected copy with its labeled path.

\begin{theorem}\label{thm:eleven}
Let a tree have positive integer capacities \(u_e\). Consider a finite collection of labeled paths, each joining two distinct leaves. Repeated paths are allowed and are treated as distinct labeled objects. If at most \(4u_e+2\) paths use each edge \(e\), then each path can be assigned one of eleven colors so that, for every edge \(e\) and every color, at most \(u_e\) paths of that color use \(e\).
\end{theorem}
For each private leaf edge, the scaled capacity in the KPP instance is four, whereas its target, unscaled capacity in Theorem~\ref{thm:eleven} is \(u_e=1\). Hence each color class contains at most one copy of any original request, because all four copies of that request use its private edges. Replacing the copies in any color class by their corresponding original requests yields a feasible subset of the original instance with the same weight. The eleven color classes have total weight at least \(4W\), so at least one has weight at least \(4W/11\). This proves Theorem~\ref{thm:eleven-gap}, assuming Theorem~\ref{thm:eleven}.

From now on, we work only with these labeled paths on the enlarged tree. We retain the notation \(u_e\) for the target edge capacities, including \(u_e=1\) on every private edge. Every such path joins two distinct leaves.

\section{The coloring argument}
\label{sec:trees}

Consider a leaf \(a\) whose incident edge has capacity two, and suppose ten paths end at \(a\). We seek a coloring in which no color appears more than twice among these paths. Place six paths in one ``bin'' and the remaining four in another, and require the paths in each bin to receive distinct colors. A color may then appear once in each bin, and hence at most twice in total.

The same idea works for any capacity. A \emph{bin at a leaf} is a subset of the paths ending there. We partition these paths into bins and require pairwise distinct colors within each bin. If the leaf edge has capacity \(u\), using at most \(u\) bins ensures that at most \(u\) paths of any one color use the edge.

We use bins of size at most six. Such bins can always be formed initially: if a leaf edge has capacity \(u\), then at most \(4u+2\le6u\) paths end at that leaf, where the inequality follows from \(u\ge1\). Order these paths arbitrarily and partition them into consecutive bins, each containing six paths except possibly the last. Thus at most \(u\) bins are needed. If no path ends at the leaf, no bin is created. Each path belongs to exactly one bin at each of its two endpoints.

We now state the precise inductive claim. It is slightly stronger than Theorem~\ref{thm:eleven}, because the bins at the leaves are part of the input and the coloring must respect them. A vertex is \emph{internal} if it has degree at least two. An edge is \emph{internal} if both endpoints are internal vertices. All other edges are incident to at least one leaf.

\begin{theorem}\label{thm:bins}
Let a tree have positive integer edge capacities \(u_e\), and let a finite collection of labeled paths, each joining two distinct leaves, be given. At each leaf, suppose the paths ending there are partitioned into prescribed bins. Assume that
\begin{enumerate}[label=\textup{(\roman*)}]
  \item every bin contains at most six paths;
  \item a leaf with incident edge \(e\) has at most \(u_e\) bins;
  \item every internal edge \(e\) is used by at most \(4u_e+2\) paths.
\end{enumerate}
There is an eleven-coloring such that the paths within each prescribed bin receive pairwise distinct colors and, for every edge \(e\), at most \(u_e\) paths of any given color use \(e\).
\end{theorem}

Theorem~\ref{thm:bins} immediately implies Theorem~\ref{thm:eleven}.
\begin{proof}[Proof of Theorem~\ref{thm:eleven}]
At each leaf, order the paths ending there arbitrarily and partition them into consecutive bins, each containing six paths except possibly the last. If the leaf edge has capacity \(u\), then at most \(4u+2\le6u\) paths end there, so at most \(u\) bins are needed. These bins satisfy conditions~(i) and~(ii) of Theorem~\ref{thm:bins}, while the load bound in Theorem~\ref{thm:eleven} gives condition~(iii). Applying Theorem~\ref{thm:bins} yields the desired eleven-coloring.
\end{proof}

The inductive proof of Theorem~\ref{thm:bins} uses a star-contraction argument similar to that of Chekuri, Mydlarz, and Shepherd~\cite[Section~2]{CMS}; we give the construction in full.

\begin{proof}[Proof of Theorem~\ref{thm:bins}]
We induct on the number of internal vertices. If there are no paths, there is nothing to color, so assume otherwise. The tree then has at least one edge.

\paragraph{Star or single-edge case.} If the tree has at most one internal vertex, it is a star or a single edge. If it is a star, fix the prescribed bins at the leaves and color the paths one at a time in any order. For a path being colored, each of its two endpoint bins contains at most five other paths. Previously colored paths in these bins therefore rule out at most ten colors in total, leaving at least one of the eleven colors available. Assign such a color and continue. At the end, the paths within each bin have pairwise distinct colors. For any leaf edge \(e\), condition~(ii) gives at most \(u_e\) bins at its leaf, and each bin contains at most one path of any given color. Hence at most \(u_e\) paths of each color use \(e\), so every edge-capacity constraint is satisfied. The same argument applies when the tree consists of a single edge, since each path still has two leaf endpoints and belongs to one bin at each endpoint.

\paragraph{Case of at least two internal vertices.} We may now assume that the tree has at least two internal vertices. These vertices induce a tree, since every vertex strictly between two internal vertices on their unique path also has degree at least two. Choose a leaf \(v\) of this induced tree. In the current tree, \(v\) has exactly one internal neighbor, denoted \(w\), and all its other neighbors are leaves. The edges from \(v\) to these leaves form a star attached to the rest of the tree through \(vw\). Figure~\ref{fig:star} shows an example with three leaves \(a,b,c\) at \(v\).

\begin{figure}[htb]
\centering
\begin{tikzpicture}[x=1cm,y=1cm,
  dot/.style={circle,fill,inner sep=1.8pt},>=Stealth]
  \node[dot,label=left:{$a$}] (a) at (0,1.1) {};
  \node[dot,label=left:{$b$}] (b) at (0,0) {};
  \node[dot,label=left:{$c$}] (c) at (0,-1.1) {};
  \node[dot,label=below:{$v$}] (v) at (1.5,0) {};
  \node[dot,label=below:{$w$}] (w) at (3,0) {};
  \node[dot,label=right:{$d$}] (d) at (4.3,0) {};
  \draw (a)--node[above] {$1$}(v);
  \draw (b)--node[above] {$1$}(v);
  \draw (c)--node[below] {$1$}(v);
  \draw[thick] (v)--node[above] {$2$}(w);
  \draw (w)--node[above] {$2$}(d);
  \draw[->] (4.9,0)--(5.8,0);
  \node[align=center,font=\small] at (5.35,-0.7)
    {remove local paths;\\contract $va,vb,vc$};
  \node[dot,label=below:{$v$}] (vv) at (6.5,0) {};
  \node[dot,label=below:{$w$}] (ww) at (8,0) {};
  \node[dot,label=right:{$d$}] (dd) at (9.3,0) {};
  \draw[thick] (vv)--node[above] {$2$}(ww);
  \draw (ww)--node[above] {$2$}(dd);
\end{tikzpicture}
\caption{A star attached through one edge. Edge labels indicate capacities for the running example. A path from \(a\) to \(b\) is local and is temporarily removed. A path from \(a\) to \(d\) crosses \(vw\) and becomes a path from \(v\) to \(d\) in the smaller tree.}
\label{fig:star}
\end{figure}
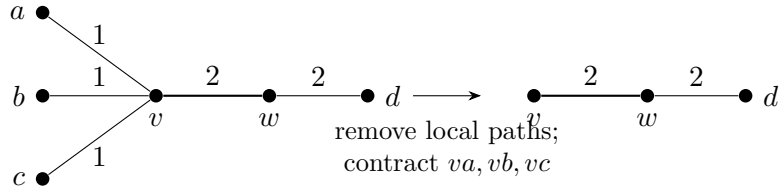

Recall that all path endpoints are leaves. A path with both endpoints among the leaves adjacent to \(v\) lies entirely within the star; call it a \emph{local path}. A path with exactly one endpoint among these leaves traverses one star edge and the edge \(vw\); call it a \emph{crossing path}.

Temporarily remove the local paths and contract the star edges, merging \(v\) and its adjacent leaves into a single vertex, again denoted \(v\). Retain the edge \(vw\) with the same capacity. After the contraction, \(v\) is a leaf and every crossing path has \(v\) as an endpoint. All other surviving paths are unchanged, and the resulting tree has fewer internal vertices.

Next, construct the bins at the new leaf \(v\). Let \(B_1,\ldots,B_k\) be the prescribed bins at the leaves adjacent to \(v\). Remove the local paths from each \(B_i\), discard any empty remainders, and relabel the nonempty crossing blocks as \(B'_1,\ldots,B'_m\). Each block has size at most six, and we treat it as an indivisible item. Every path using \(vw\) is a crossing path, and every crossing path belongs to exactly one block; hence the total size of the blocks equals the load on \(vw\). This load is at most \(4u_{vw}+2\) by condition~(iii), since \(vw\) was internal before the contraction. To apply the induction hypothesis, we must pack these items into at most \(u_{vw}\) new bins, each of total size at most six. Lemma~\ref{lem:packing} shows that this is possible: begin with one new bin for each crossing block and repeatedly merge two bins whose combined size is at most six until no such pair remains. At every other leaf, retain the prescribed bins.

We now verify the induction hypotheses for the smaller tree. By Lemma~\ref{lem:packing}, each new bin at \(v\) has size at most six, and there are at most \(u_{vw}\) such bins. Thus they satisfy conditions~(i) and~(ii). The bins at every other leaf are unchanged. Every edge that remains internal was internal before the contraction, and its load is unchanged because the removed local paths used only the contracted star edges. Thus condition~(iii) continues to hold. Every surviving path still joins two distinct leaves: a crossing path now joins \(v\) to its original endpoint outside the star. Finally, the contraction turns \(v\) from an internal vertex into a leaf and creates no new internal vertex. The smaller tree therefore has fewer internal vertices, so the induction hypothesis supplies the required coloring.

Undo the contraction while retaining the colors of the surviving paths. In each prescribed bin at a leaf adjacent to \(v\), the crossing paths already have pairwise distinct colors. Indeed, the corresponding crossing block was kept intact within a single new bin, whose paths received distinct colors in the recursive coloring.

Now reinsert the local paths one at a time. Each local path belongs to one prescribed bin at each of its two endpoints, and each of these bins contains at most five other paths. The already colored paths in the two bins therefore exclude at most ten colors, leaving at least one of the eleven colors available. Assign such a color and continue. After all local paths have been inserted, the paths in every prescribed bin have pairwise distinct colors.

It remains to verify the edge capacities. For each restored star edge \(e\), condition~(ii) gives at most \(u_e\) prescribed bins at its leaf, and each bin contains at most one path of any given color. Thus at most \(u_e\) paths of each color use \(e\). On every retained edge, including \(vw\), the load of each color class is unchanged from the recursively colored instance because local paths use only the restored star edges. Hence every edge-capacity constraint is satisfied, completing the induction step.
\end{proof}

\section{The algorithm}
\label{sec:algorithm}

The preceding proofs yield the following algorithm. Solve the LP, add the two private leaves for each original request, and create four copies of each extended request. Apply the rounding algorithm of~\cite{KPP} after multiplying every edge capacity by four. Regard the selected copies as the collection of paths to be colored, form the initial bins at the leaves, and repeatedly contract an attached star. At each step, record the pre-contraction bins and the local paths that are removed. When the remaining tree is a star or a single edge, greedily color its remaining paths. Then undo the contractions in reverse order, coloring the local paths as they are restored. Finally, choose a maximum-weight color class and replace each copy in it with the corresponding original request.

All these operations can be performed in polynomial time. The enlarged tree adds only two leaves per original request, and the copied instance contains only four copies of each request. Each contraction reduces the number of internal vertices. In each application of the packing lemma, every merge reduces the number of nonempty bins. Path lists, bin memberships, and sets of available colors can be maintained explicitly in polynomial time using polynomial space. Because only nonempty bins are stored, the representation does not grow in proportion to the numerical values of the capacities. This establishes the algorithmic claim in Theorem~\ref{thm:eleven-gap}.

\section{Saving the eleventh color}
\label{sec:ten-colors}

When we restored a contracted star in the recursion, its crossing paths already had colors, and we colored the local paths one at a time. Each local path belonged to two bins, and each bin contained at most five other paths. The colors already used in these two bins could therefore exclude as many as $5+5=10$ colors. With eleven colors available, there was always a color left. With ten colors, this greedy step may get stuck.

Some of the colors blocking a path, however, were assigned to other local paths. Those choices can be changed. Only the crossing paths have colors that we must keep: their colors came from the recursive call, and changing them could create a conflict elsewhere in the tree. Thus we have more freedom if we choose the colors of all the local paths together. The following consequence of an edge-coloring theorem of Edwards et al.~\cite{EGHKPS} makes this useful.

\begin{lemma}\label{lem:bin-completion}
Consider the bins at the leaves of an attached star. Every bin contains at most six paths. Suppose the crossing paths have already been assigned colors from a set of ten colors, with distinct colors within each bin. We can color all the local paths using the same ten colors, without changing any crossing color, provided that either
\begin{enumerate}[label=\textup{(\alph*)}]
  \item every bin contains at most four crossing paths; or
  \item there are at most six crossing paths altogether, and their colors are all different.
\end{enumerate}
The resulting coloring gives distinct colors to the paths in each bin. We can do this in polynomial time.
\end{lemma}

The proof of the lemma is deferred to Appendix~\ref{sec:specialization}, but we note that it is essentially a specialization of the edge-coloring theorem of Edwards et al.~\cite{EGHKPS}. 

\begin{theorem}\label{thm:ten}
Under the assumptions of Theorem~\ref{thm:bins}, ten colors suffice: the paths in each prescribed bin receive distinct colors, and each edge \(e\) is used by at most \(u_e\) paths of any one color.
\end{theorem}

\begin{proof}
If there are no paths, there is nothing to color. If the tree is a star or a single edge, then all paths are local. So there are then no crossing paths, so condition~\textup{(b)} of Lemma~\ref{lem:bin-completion} holds vacuously.

For a larger tree, choose an attached star exactly as in Section~\ref{sec:trees}. Denote its center by \(v\), its only internal neighbor by \(w\), and the capacity of \(vw\) by \(u\). Recall that local paths join two leaves adjacent to \(v\), whereas crossing paths have one endpoint at such a leaf and leave the star through \(vw\).

\paragraph{When the same contraction works.} Suppose first that every bin at a leaf adjacent to \(v\) contains at most four crossing paths. Remove the local paths and contract the leaf edges at \(v\), just as before. Keep the crossing paths from each old bin together. Their total number is at most \(4u+2\), so Lemma~\ref{lem:packing} packs these groups into at most \(u\) new bins, each of size at most six. The resulting tree satisfies the same three conditions as in Section~\ref{sec:trees}; color it recursively with ten colors.

Restore the star, keeping the colors of the crossing paths. Crossing paths from the same old bin still have distinct colors, because they stayed together in a new bin. Each old bin has at most four such paths, so Lemma~\ref{lem:bin-completion}\textup{(a)} colors the local paths without changing those colors. This completes the coloring of the star. The bins enforce the capacities on its leaf edges, and the local paths use no other edges.

The same contraction also works whenever \(u=1\), even if a bin contains five or six crossing paths. In this case at most \(4\cdot1+2=6\) paths cross \(vw\), and the contraction puts all of them in a single bin. Consequently, the recursive coloring gives them all different colors. Lemma~\ref{lem:bin-completion}\textup{(b)} then colors the local paths, and the same capacity check applies.

\paragraph{When a bin contains at least five crossing paths.} The only case left is \(u\ge2\) with a bin at a leaf adjacent to \(v\) containing at least five crossing paths. Let \(a\) be the leaf carrying this bin, and call the bin \(B\). It has room for at most one local path, since a bin contains at most six paths altogether.

Attach a new leaf \(b\) to \(w\), on the other side of the connecting edge. Replace the endpoint \(a\) of every path in \(B\) by \(b\); keep its other endpoint and its label. These paths now follow the unique routes from \(b\) to their unchanged endpoints. Paths outside \(B\) are unchanged, including paths in other bins at \(a\).

Consider what this does to a crossing path in \(B\). Previously, it traveled from \(a\) through \(v\) and \(w\) before continuing to its other endpoint. Now it starts at \(b\), reaches \(w\) immediately, and follows the same route from there. It no longer uses \(vw\). A local path in \(B\) changes in the opposite way: its other endpoint is a leaf at \(v\), so its new route from \(b\) uses \(vw\). There are at least five paths of the first kind and at most one of the second. Thus the new load on \(vw\) is at most
\begin{align*}
  \text{old load on }vw-5+1
    &\le (4u+2)-5+1\\
    &=4u-2\\
    &=4(u-1)+2.
\end{align*}
We can therefore lower the capacity of \(vw\) from \(u\) to \(u-1\) and retain the required bound on its total load. Because \(u\ge2\), the new capacity is still positive.

Give the new edge \(bw\) capacity one, and place all paths of \(B\) in one bin at \(b\). At \(a\), remove \(B\) from the old list of bins and lower the capacity of \(av\) by one. The number of bins at \(a\) was at most \(u_{av}\); after removing \(B\), it is at most \(u_{av}-1\), exactly the new capacity of \(av\). If \(av\) had capacity one, \(B\) was its only bin; no path now uses \(av\), so delete this edge and the leaf \(a\). Figure~\ref{fig:one-bin-move} shows the change.

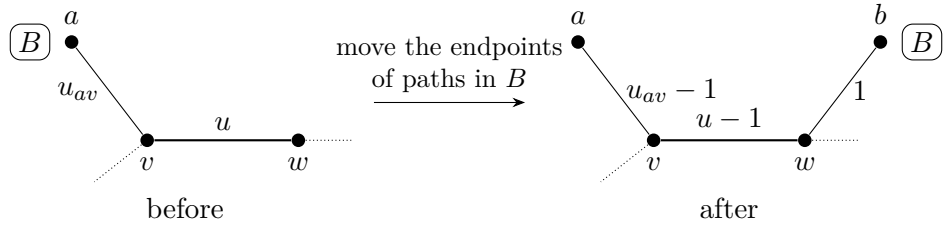
\begin{figure}[htb]
\centering
\begin{tikzpicture}[x=1cm,y=1cm,
  dot/.style={circle,fill,inner sep=1.8pt},>=Stealth]
  \node[dot,label=above:{$a$}] (a) at (0,1.3) {};
  \node[draw,rounded corners,inner sep=3pt] at (-0.55,1.3) {$B$};
  \node[dot,label=below:{$v$}] (v) at (1,0) {};
  \node[dot,label=below:{$w$}] (w) at (3,0) {};
  \draw (a)--node[left] {$u_{av}$}(v);
  \draw[thick] (v)--node[above] {$u$}(w);
  \draw[densely dotted] (v)--(0.3,-0.55);
  \draw[densely dotted] (w)--(3.7,0);
  \node at (1.5,-0.9) {before};
  \draw[->] (4,0.5)--node[above,align=center,font=\small]
    {move the endpoints\\of paths in $B$}(6,0.5);
  \node[dot,label=above:{$a$}] (aa) at (6.7,1.3) {};
  \node[dot,label=below:{$v$}] (vv) at (7.7,0) {};
  \node[dot,label=below:{$w$}] (ww) at (9.7,0) {};
  \node[dot,label=above:{$b$}] (b) at (10.7,1.3) {};
  \node[draw,rounded corners,inner sep=3pt] at (11.25,1.3) {$B$};
  \draw (aa)--node[right] {$u_{av}-1$}(vv);
  \draw[thick] (vv)--node[above] {$u-1$}(ww);
  \draw (ww)--node[right] {$1$}(b);
  \draw[densely dotted] (vv)--(7,-0.55);
  \draw[densely dotted] (ww)--(10.4,0);
  \node at (8.7,-0.9) {after};
\end{tikzpicture}
\caption{Moving one bin. Paths in \(B\) acquire the new endpoint \(b\); other bins at \(a\) stay there. Only the affected edges are shown, with their capacities. If \(u_{av}=1\), delete \(av\) and \(a\) after the move.}
\label{fig:one-bin-move}
\end{figure}

Every other retained internal edge carries exactly the same paths as before, so its load continues to satisfy the required bound. Every path still joins distinct leaves. If deleting \(a\) makes \(v\) a leaf, no path ends there, and it has no bins. Thus the new instance satisfies the same three conditions for the bins and loads.

Color this instance recursively with ten colors. Then put the endpoints of the paths in \(B\) back at \(a\), restore the two old capacities, and keep all the colors. Delete \(b\) and \(bw\), which are now unused. Because the paths of \(B\) shared the new capacity-one edge \(bw\), they have distinct colors. Hence restoring them adds at most one path of each color to \(av\), and at most one to \(vw\). For any fixed color, the loads on these two edges are at most
\[
  (u_{av}-1)+1=u_{av}
  \qquad\text{and}\qquad
  (u-1)+1=u,
\]
respectively. If \(av\) was deleted, its load before restoration is zero and the first equality gives the same bound. The paths of \(B\) still have distinct colors when returned to their old bin. All other bins are unchanged, and no other old edge gains or loses a path. This gives the required coloring of the original tree.

\paragraph{Why the recursion terminates.} The contraction removes tree edges and never lengthens a path. The endpoint change may add one tree edge, but it makes at least five paths shorter by one edge and at most one path longer by one edge. The total path length therefore falls by at least four. In either case, the nonnegative integer
\[
  \text{number of tree edges}
  +\sum_{\text{paths }f}\text{number of edges in }\Path(f)
\]
strictly decreases. Formally, we use induction on this sum.

This also gives a polynomial-time algorithm. If the initial tree has \(m\) edges and there are \(p\) paths, the sum is at most \(m(p+1)\), so there are at most \(m(p+1)\) recursive steps. Each step takes polynomial time, including the coloring completions supplied by Lemma~\ref{lem:bin-completion}.
\end{proof}

\appendix 
\section{Packing lemma}
\label{sec:packing}

\begin{lemma}\label{lem:packing}
Let \(u\ge1\) be an integer. Any collection of indivisible items with positive integer sizes at most six and total size at most \(4u+2\) can be packed into at most \(u\) bins, each having total size at most six, without splitting any item.
\end{lemma}

\begin{proof}
Initially, place each item in its own bin. While there is a pair of bins whose combined size is at most six, merge such a pair. Every bin continues to have size at most six, and no item is split. The process terminates because each merge reduces the number of bins by one.

Let \(q\) be the number of bins at termination. If \(q\le1\), then \(q\le u\); hence assume \(q\ge2\). Denote the bin sizes by \(s_1,\ldots,s_q\), and let \(L=s_1+\cdots+s_q\) be their total size. No two remaining bins can be merged. Since the bin sizes are integers,
\[
  s_i+s_j\ge7
  \qquad\text{whenever }1\le i<j\le q.
\]
Let \(s=\min_i s_i\) be the smallest remaining bin size. If \(s\ge4\), then
\[
  L\ge4q\ge4q-1.
\]
If instead \(s\le3\), the smallest bin cannot be merged with any other bin, so every other bin has size at least \(7-s\). Therefore
\begin{align*}
  L
    &\ge s+(q-1)(7-s)\\
    &\ge 4q-1,
\end{align*}
where the last inequality follows from
\[
  s+(q-1)(7-s)-(4q-1)=(q-2)(3-s)\ge0.
\]
Thus \(L\ge4q-1\) in both cases. If \(q>u\), the integrality of \(q\) and \(u\) implies \(q\ge u+1\), and hence
\[
  L\ge4q-1\ge4(u+1)-1=4u+3.
\]
This contradicts the assumed bound \(L\le4u+2\). Therefore \(q\le u\), as desired.
\end{proof}

\section{Specialization of the edge-coloring theorem}
\label{sec:specialization}
\begin{proof}[Proof of Lemma~\ref{lem:bin-completion}]
This is a specialization of Edwards, Gir\~ao, van den Heuvel, Kang, Puleo, and Sereni~\cite[Theorem~9]{EGHKPS}. Represent bins by vertices and local paths by edges. In~\textup{(a)}, represent each crossing path by a precolored edge to a separate new leaf. The graph has maximum degree at most six and at most four precolored edges per vertex, so the cited theorem completes the coloring. In~\textup{(b)}, represent the crossing paths by edges to one common new vertex. The maximum degree is still at most six, so the same theorem colors this graph from scratch. The crossing edges receive distinct colors at the common vertex; permuting the colors throughout the graph matches them to their distinct prescribed colors.

\paragraph{Running time.} The edge-coloring steps can be performed in polynomial time as well. Chleb\'ik and Chleb\'ikov\'a~\cite[Theorem~3]{CC} give an algorithm for extending a partial vertex coloring when the number of available colors is at least the maximum degree of the graph. To apply it here, form a graph with one vertex for each edge of our multigraph, joining two vertices when their edges share an endpoint. Each edge of the multigraph meets at most \(5+5=10\) other edges, so this new graph has maximum degree at most ten. Its vertex colorings are precisely the proper edge colorings we need. The algorithm finds the extension whose existence was established above.
\end{proof}

\bibliographystyle{apalike} 
\bibliography{bibliography} 

\end{document}